\documentclass[runningheads]{llncs}
\usepackage{amssymb}

\usepackage[T1]{fontenc}
\usepackage{graphicx}
\begin{document}
\title{On Secret Sharing from Extended Norm-Trace Curves}
%
%
\author{Olav Geil\orcidID{0000-0002-9666-3399}}
\authorrunning{O.\ Geil}
\titlerunning{On Secret Sharing from Extended Norm-Trace Curves}
%
\institute{Department of Mathematical Sciences, Aalborg University, Denmark 
\email{olav@math.aau.dk}}
\maketitle              
\begin{abstract}
In~\cite{camps2025weight} Camps-Moreno et al. treated (relative) generalized Hamming weights of codes from extended norm-trace curves and they gave examples of resulting good asymmetric quantum error-correcting codes employing information on the relative distances. In the present paper we study ramp secret sharing schemes which are objects that require an analysis of higher relative weights and we show that not only do schemes defined from one-point algebraic geometric codes from extended norm-trace curves have good parameters, they also possess a second layer of security along the lines of~\cite{geil2026considerate}. It is left undecided in~\cite[page 2889]{camps2025weight} if the ``footprint-like approach'' as employed by Camps-Moreno herein is strictly better for codes related to extended norm-trace codes than the general approach for treating one-point algebraic geometric codes and their likes as presented in~\cite{geil2014relative}.  
We demonstrate that the method used in~\cite{camps2025weight} to estimate (relative) generalized Hamming weights of codes from extended norm-trace curves can be viewed as a clever application of the enhanced Goppa bound in~\cite{geil2014relative} rather than a competing approach. \\

\end{abstract}

\section{Introduction}\label{secintro}
Secret sharing is a branch of cryptography in which a central dealer shares a secret among a set of participants by giving them each a share. This is done in such a way that if few of them pool their shares then they obtain no information on the secret, but if many do then they can recover it in full. One often requires the schemes to be linear meaning that (partial) recovery, whenever possible, can be done by simple and fast linear algebra algorithms.  In most schemes studied in the literature the shares and the secrets are elements from the same set, often a finite field ${\mathbb{F}}_q$. This is in contrast to linear ramp secret sharing where the secret belongs to ${\mathbb{F}}_q^\ell$, and the shares to ${\mathbb{F}}_q$, with $\ell \geq 1$. 
Ramp schemes were originally introduced by Blakley and Meadows in~\cite{Blakley} and by Yamamoto in~\cite{Yamamoto}, and have gained a lot of interest due to their rich structures and their application in connection with for instance storage of bulk data and secure multiparty computation~\cite{ChenCramer,csirmaz2009ramp}. A linear ramp scheme is synonymous to a pair of nested linear codes and worst case (partial) information leakage as well as worst case (partial) information recovery can be described by relative parameters of the pair of codes and their duals. 

For large number of participants compared to the field size knowledge on the worst case scenarios, however, by far gives the complete picture, but a deeper analysis is typically very difficult. In~\cite{geil2026considerate} the author coined the concept of maximum non-$i$-qualifying sets to address this difficulty. These are sets of maximum size being not able to recover $i$ $\log_2(q)$ bits of information. It was demonstrated that two specific classes of monomial-Cartesian codes possess families of maximum non-$i$-qualifying sets with a very rich structure and based on this information one can then impose a second layer of security. By elaborating on results in~\cite{camps2025weight} in the present paper we demonstrate that also ramp secret sharing schemes constructed from nested one-point algebraic geometric codes from the extended norm-trace curves posses families of well-structured maximum non-$i$-qualifying sets giving again rise to a second layer of security. 

The exposition in~\cite{camps2025weight} uses the language of Gröbner basis theory, but as we demonstrate their strategy for deriving estimates on code parameters related to the extended norm-trace curve is not a competing method to the enhanced Goppa kind of bounds in~\cite[Eq.\ 13, Eq.\ 17]{geil2014relative} for general algebraic function fields of transcendence degree $1$, but can be viewed as a clever application thereof. By proving this we answer an open question raised by Camps-Moreno et al. in~\cite[page 2889]{camps2025weight}.

The paper is organized as follows. In Section~\ref{seclinramp} we give the necessary background on linear ramp secret sharing. We next treat general theory of relative generalized Hamming weights of nested one-point codes and their relatives in Section~\ref{secgeneral}.  Then in Section~\ref{secdecnormtrace} we elaborate on material in~\cite{camps2025weight} regarding (relative) generalized Hamming weights of nested decreasing norm-trace codes and their duals, and demonstrate the relationship between the method used in~\cite{camps2025weight} and that of~\cite{geil2014relative}.  Finally, in Section~\ref{secsecondlayer} we apply such results to produce the alluded linear ramp schemes with desirable structure and consequently a second layer of security. Section~\ref{secconclusion} contains concluding remarks.

\section{Linear Ramp Secret Sharing}\label{seclinramp}
In linear ramp secret sharing the shares given to the individual participants belong to a finite field $\mathbb{F}_q$, but the secret is a tuple in ${\mathbb{F}}_q^\ell$, where $\ell$ is allowed to be strictly larger than $1$. The linearity means that if $c^{(1)}_1, \ldots , c^{(1)}_n $ are shares for a secret $\vec{s}^{(1)}$ and $c^{(2)}_1, \ldots , c^{(2)}_n $ are shares for a secret $\vec{s}^{(2)}$ then also $ac^{(1)}_1+bc_1^{(2)}, \ldots , ac^{(1)}_n+bc_n^{(2)} $ work as shares for the secret $a\vec{s}^{(1)}+b\vec{s}^{(2)}$ for any $a,b \in {\mathbb{F}}_q$. Such schemes by~\cite[Sec.\ 4.2]{ChenCramer} can be put into the form of a nested code construction as follows. 
Let $C_2 \subseteq C_1 \subseteq {\mathbb{F}}_q^n$, $\dim C_2 = k_2 < \dim C_1 =k_1$ and $\ell = k_1-k_2 $. Consider bases $\{\vec{b}_1, \ldots , \vec{b}_{k_2}\}$ and  $\{\vec{b}_1, \ldots , \vec{b}_{k_2}, \vec{b}_{k_2+1}, \ldots , \vec{b}_{k_1}\}$, respectively, for $C_2$ and $C_1$, respectively, as vector spaces over ${\mathbb{F}}_q$. A secret $\vec{s}=(s_1, \ldots , s_{\ell})  \in {\mathbb{F}}_q^\ell$ is encoded to 
\begin{equation}
(c_1, \ldots , c_n)= r_1\vec{b}_1 + \cdots + r_{k_2} \vec{b}_{k_2}+s_1\vec{b}_{k_2+1} + \cdots + s_\ell \vec{b}_{k_1}\label{eqsnabel1}
\end{equation}
and participant $i$, then receives $c_i$ as a share $i=1, \ldots , n$.\\
The crucial parameters of a linear ramp secret sharing scheme are the number of participants, which is $n$, the dimension of space of secrets, which is $\ell$, the privacy numbers $t_1, \ldots , t_\ell$, and the reconstruction numbers $r_1, \ldots , r_\ell$. Here, for $i=1, \ldots , \ell$, $t_i$ is largest possible and $r_i$ is smallest possible such that 
\begin{itemize}
\item no set of $t_i$ participants is able to recover $i$ times $\log_2(q)$ bits of information about $\vec{s}$
\item any set of $r_i$ participants can recover $i$ times $\log_2(q)$ bits of information.
\end{itemize} 
Of special interest are full privacy and full recovery corresponding to $t=t_1$ and $r=r_\ell$. 
As is well-known~\cite{geil2014relative,kurihara}
\begin{eqnarray}
t_m&=&M_m(C_2^\perp, C_1^\perp)-1 \label{eqtm} \\
r_m&=&n-M_{\ell-m+1}(C_1,C_2)+1 \label{eqrm}
\end{eqnarray}
where 
\begin{eqnarray}
M_t(C_1,C_2)&=&\min \{ \# {\mbox{Supp}}(D) \mid D \subseteq C_1, D \cap C_2=\{\vec{0} \}, \dim D =t\} \nonumber
\end{eqnarray}
is called the $t$th relative generalized Hamming weight. 

As is always the case in coding theory, one cannot choose parameters freely. If, say, $q, n, \ell$ and $1\leq r_1 < \cdots <r_\ell \leq n$ are fixed numbers then among the set of all pairs of nested codes $C_2 \subseteq C_1 \subseteq {\mathbb{F}}_q^n$ with such parameters there are limitations as to how close $t_i$ can be to $r_i$, $i=1, \ldots , \ell$, see~\cite{cascudo2013bounds}.  However, by definition for any $i$ there exist groups of size $r_i-1$ who cannot recover $i$ times $\log_2(q)$ bits of information and when we have knowledge about the structure of such groups we may employ it as a second layer of security as demonstrated in the recent paper~\cite{geil2026considerate}. The importance of the mentioned groups justify that we give them a particular name~\cite[Def.\ 8]{geil2026considerate}.
\begin{definition}
Given a linear ramp secret sharing scheme defined from $C_2 \subseteq C_1 \subseteq {\mathbb{F}}_q^n$, a set $A \subseteq \{1, \ldots , n\}$ is called {\textit{maximum non-$i$-qualifying}} if $\# A =n-M_{\ell-i+1}(C_1,C_2)$, but from the corresponding shares one cannot recover $i \log_2(q)$ bits of information.
\end{definition}
To establish information on the access-structure (who can recover how much) and in particular information on the maximum non-$i$-qualifying sets we may employ the following result~\cite[Thm.\ 3]{geil2026considerate} which is a slight reformulation of~\cite[Thm.\ 10]{ChenCramer} (see also~\cite{patel2024efficient}).
\begin{theorem}\label{thethree} 
Let $A=\{ i_1 < \cdots < i_m\} \subseteq \{1, \ldots , n\}$. Assume $c_{i_1}^\prime , \ldots , c_{i_m}^\prime$ are simultaneously realizable shares in those positions, see~(\ref{eqsnabel1}). The amount of possible secrets $\vec{s}$ corresponding to such shares equals $q^s$ with
$$
s= \max \{\dim D \mid D \subseteq C_1, D \cap C_2 = \{\vec{0} \}, {\mbox{Supp}}(D) \subseteq \bar{A} \}
$$
where $\bar{A}=\{1, \ldots , n\} \backslash A$.
\end{theorem}
In~\cite{geil2026considerate} Theorem~\ref{thethree} was employed to show that ramp schemes coming from the particular nested monomial-Cartesian codes as described in~\cite[Sec.\ 4]{MR3782266} does not only have good parameters $n, \ell, r_1, \ldots , r_\ell, r_1, \ldots , t_\ell$, but also support a second layer of security by possessing maximum non-$i$-qualifying sets of desirable well-structured form. This insight was then used in~\cite[Sec.\ 4]{geil2026considerate} to construct monomial-Cartesian code based schemes with larger $\ell$ maintaining desirable second layer of security by possessing well-structured maximum non-$i$-qualifying sets whilst paying less interest in the worst case security $t_1, \ldots , t_\ell$. The aim of the present contribution is to investigate what can be said regarding a second layer of security for ramp schemes based on decreasing norm-trace codes.

\section{Relative Parameters of One-Point Algebraic Geometric Codes and Their Relatives}\label{secgeneral}
Before treating codes from extended norm-trace curves we revisit the general theory of relative generalized Hamming weights of any pair of nested one-point algebraic geometric codes and their relatives~\cite{geil2014relative}. 
Consider an arbitrary algebraic function field of transcendence degree 1 defined over some finite field ${\mathbb{F}}_{\mathcal{Q}}$. Let $P_1, \ldots , P_n, Q$ be pairwise different rational places. Let $H(Q)$ be the Weierstrass semigroup of $Q$, i.e.\ 
$$H(Q)=-\nu_Q( R=\cup_{m=0}^\infty {\mathcal{L}}(mQ))$$
where $\nu_Q$ is the discrete valuation corresponding to $Q$. Define
$$H^\ast(Q)=\{ \lambda \in H(Q) \mid C_{\mathcal{L}}(G=P_1+\cdots +P_n,\lambda Q) \neq C_{\mathcal{L}}(G,(\lambda -1) Q) \}.$$
Clearly, $\# H^\ast (Q)=n$. 
We now recall material from~\cite{geil2014relative} on how to estimate (relative) generalized Hamming weights. Let $D \subseteq {\mathbb{F}}_{\mathcal{Q}}^n$ be any subspace, of dimension say $m$.  There exist functions $f_1, \ldots , f_m \in R$ such that 
$$\{ (f_i(P_1), \ldots , f_i(P_n)) \mid i=1, \ldots , m\}$$
is a basis for $D$ and such that 
\begin{equation}
-\nu_Q(f_1) < \cdots < -\nu_Q(f_m) \label{eqnuQ}
\end{equation} 
holds, where without loss of generality we may assume that all numbers in~(\ref{eqnuQ}) belong to $H^\ast(Q)$. 
These values do not depend on the choice of the $f_i$'s, but are invariants of $D$ and we may therefore define $\rho(D)=\{ -\nu_Q(f_1), \ldots , -\nu_Q(f_m)\}\subseteq H^\ast(Q)$. From~\cite[Prop.\ 17]{geil2014relative} we have the bound 
\begin{eqnarray}
\# {\mbox{Supp}} (D)&\geq & \# \big( H^\ast (Q) \cap \big( \cup_{s=1}^m (\gamma_s +H(Q))\big) \big) \label{eqestimateD}
\end{eqnarray}
where $\rho(D) =\{ \gamma_1, \ldots , \gamma_m\} $. 

With proper care one can apply~(\ref{eqestimateD}) to any set of nested codes defined from $R$, but for nested one-point codes the situation is immediate~\cite[Thm.\ 19]{geil2014relative}. Let $\lambda_2< \lambda_1$ be elements in $H^\ast(Q)$ and consider $m \leq \dim C_{\mathcal{L}}(G,\lambda_1 Q) - \dim C_{\mathcal{L}}(G,\lambda_2 Q)$. We have 
\begin{eqnarray}
&&M_m( C_{\mathcal{L}}(G,\lambda_1Q),C_{\mathcal{L}}(G,\lambda_2Q))\nonumber \\
&\geq &\min \bigg\{ \# \bigg(H^\ast (Q) \cap \big( \cup_{s=1}^m (\gamma_s+H(Q))\big) \bigg)\mid \gamma_s \in H^\ast (Q) {\mbox{ for }} s=1, \ldots , m,\nonumber \\
&& {\mbox{ \ \hspace{6cm} \ }} \lambda_2 < \gamma_1 < \cdots <\gamma_m \leq \lambda_1\bigg\}.\label{eqsnabel47}
\end{eqnarray}
That this bound (and in larger generality (\ref{eqestimateD})) can be seen as an enhancement of the Goppa bound for a one-point Goppa code i.e.\ the bound
$$
d(C_{\mathcal{L}}(G,\lambda Q)) \geq n-\lambda 
$$
follows from ~\cite[Lem.\ 5.15]{handbook} which says that for any numerical semigroup $\Lambda$ and any element $\lambda$ herein we have $\lambda=\# (\Lambda \backslash (\lambda + \Lambda))$. In fact~(\ref{eqsnabel47}) can be a strict improvement even when applied to estimate the minimum distance of a one-point code. This happens when some of the elements in $\Lambda\backslash (\lambda + \Lambda)$ are not contained in $H^\ast (Q)$.  

As noted in~\cite{geil2014relative} a translation of~\cite[Thm\ 14]{geil2014relative} into the particular case of one-point algebraic geometric codes and their relatives produces a similar result as~(\ref{eqestimateD}), but for dual spaces. We now fill in the missing details by stating such result. Consider a subspace $D \subseteq {\mathbb{F}}_{\mathcal{Q}}^n$ of dimension $m$ and let $\eta_1, \ldots , \eta_m$ be the unique numbers such that for $i=1, \ldots , m$ there exists some $\vec{c}\in D$ satisfying $\vec{c} \in C_{\mathcal{L}}^\perp (G,(\eta_i-1)Q)$, but $\vec{c} \notin C_{\mathcal{L}}^\perp (G, \eta_i Q)$. We shall write $\kappa (D)=\{ \eta_1, \ldots , \eta_m\} \subseteq H^\ast(Q)$. The counterpart to~(\ref{eqestimateD}) is 
\begin{eqnarray}
\# {\mbox{Supp}} (D)&\geq & \# \big( H(Q) \cap \big( \cup_{s=1}^m (\eta_s -H(Q))\big) \big) \label{eqestimateDdual}
\end{eqnarray}
From this one immediately obtains \cite[Thm.\ 20]{geil2014relative} which we now state
\begin{eqnarray}
&&M_m( C^\perp_{\mathcal{L}}(G,\lambda_2Q),C^\perp_{\mathcal{L}}(G,\lambda_1Q))\nonumber \\
&\geq &\min \bigg\{ \# \bigg(H(Q) \cap \big( \cup_{s=1}^m (\gamma_s-H(Q))\big) \bigg)\mid \gamma_s \in H^\ast (Q) {\mbox{ for }} s=1, \ldots , m,\nonumber \\
&& {\mbox{ \ \hspace{6cm} \ }} \lambda_2 < \gamma_1 < \cdots <\gamma_m \leq \lambda_1\bigg\}.\label{eqsnabel47pkt5}
\end{eqnarray}
(Note that in the above formula $\gamma_1, \ldots , \gamma_m$  play the role of $\eta_1, \ldots , \eta_m$). 
\section{Bounds on Relative Parameters of Codes from the Extended Norm-Trace Curve}\label{secdecnormtrace}

The extended norm-trace curve is the curve over ${\mathbb{F}}_{q^s}$ given by the equation
\begin{equation}
x^u =y^{q^{s-1}}+y^{q^{s-2}}+\cdots +y, \label{eqnt}
\end{equation}
where $u$ is a positive divisor of $\frac{q^s-1}{q-1}$. Recall, that the right hand side of~(\ref{eqnt}) corresponds to the trace map from ${\mathbb{F}}_{q^s}$ to ${\mathbb{F}}_q$, and similarly that $x^{\frac{q^s-1}{q-1}}$ corresponds to the norm map. Related codes are therefore generalizations of norm-trace codes~\cite{normtrace,matthews2022norm} which again are generalizations of Hermitian codes~\cite{MR1737936,stichtenothhermitian,yanghermitian}. Codes from the extended norm-trace curve have been extensively studied in~\cite{bras2008duality,camps2025weight,carvalho2024decreasing,janwa2016parameters} and one of the crucial observations employed here is the systematic structure of the affine variety, i.e. the set of affine roots of~(\ref{eqnt}).  As is easily seen from the property of the norm-map and the trace-map the affine point set is the disjoint union of the following sets~\cite[Lem.\ 3.1]{carvalho2024decreasing},
$$
A_0=\{(0,b) \mid b^{q^{s-1}}+\cdots +b=0\} =\Gamma_0^{(1)} \times \Gamma_0^{(2)}
$$
and for $i=1, \ldots , q-1$
$$A_i=\{(a,b) \mid a^u=\alpha^i, b^{q^{s-1}}+\cdots +b=\alpha^i \} =\Gamma_i^{(1)} \times \Gamma_i^{(2)},$$
where $\alpha$ is a primitive element of ${\mathbb{F}}_q$. Clearly, $\# A_0=\# \Gamma_0^{(2)}=q^{s-1}$ and for $i=1, \ldots , q-1$ we have $\#A_i = u q^{s-1}$ where $\# \Gamma_i^{(1)}=u$ and $\# \Gamma_i^{(2)}=q^{s-1}$, \cite[Lem.\ 3.1]{carvalho2024decreasing}.
The algebraic function field of transcendence degree $1$ over ${\mathbb{F}}_{q^s}$ defined from~(\ref{eqnt}) has exactly $1+\sum_{i=0}^{q-1} \# A_i$ rational places \cite[Sec.\ 4.2]{carvalho2024decreasing}, each of them, but one, being related to an affine point, and the last being the unique place at infinity. Following \cite[Sec.\ 4.2]{carvalho2024decreasing} and \cite{pellikaanexists} the Weierstrass semigroup corresponding to the unique place $Q$ at infinity equals $$H(Q)=\langle u, q^{s-1} \rangle$$ and we have $$R=\cup_{m=0}^\infty {\mathcal{L}} (mQ)={\mathbb{F}}_{q^s}[X,Y]/I$$ where $I=\langle X^u-Y^{q^{s-1}}-\cdots -Y\rangle$.
Writing $$x^iy^j=X^iY^j+\langle X^u-Y^{q^{s-1}}-\cdots -Y\rangle$$ and following~\cite[Sec.\ 4.2]{carvalho2024decreasing} one sees that $$\{x^iy^j \mid 0 \leq i , 0 \leq j < q^{s-1} \}$$ is a basis for $R$ as a vector space over ${\mathbb{F}}_{q^s}$ and clearly $-\nu_Q(x^iy^j)=i q^{s-1}+ju$, where $\nu_Q$ indicates the discrete valuation corresponding to $Q$. No two elements of this basis have the same $\nu_Q$-value, i.e.\ there is a one-to-one correspondence between the basis and $H(Q)$. 

We shall denote by $P_1, \ldots , P_n$ the rational places different from $Q$, where of course $n=(u(q-1)+1)q^{s-1}$. Writing $G=P_1+\cdots +P_n$ and as before 
$$H^\ast (Q)=\{ \lambda \in H(Q) \mid C_{\mathcal{L}}(G,\lambda Q) \neq C_{\mathcal{L}}(G,(\lambda -1)Q)\},$$
it is not difficult to see that $$H^\ast (Q)=\{ i q^{s-1}+j u \mid 0\leq i < u(q-1)+1, 0 \leq j <q^{s-1}\}$$ and that $$\{(f(P_1), \ldots , f(P_n)) \mid f=x^iy^j, {\mbox{ where }} 0 \leq i < u(q-1)+1, 0 \leq j < q^{s-1}\}$$ is a basis for ${\mathbb{F}}_{q^s}^n$ as a vector space over ${\mathbb{F}}_{q^s}$, and in particular that 
\begin{eqnarray}
\{(f(P_1), \ldots , f(P_n)) \mid  f=x^iy^j, {\mbox{ where }} 0 \leq i < u(q-1)+1, 0 \leq j < q^{s-1}, {\mbox{ \ \ \ \ }} \nonumber \\
 iq^{s-1}+ju\leq \lambda\} \nonumber 
\end{eqnarray} is a basis for $C_{\mathcal{L}}(G,\lambda Q)$. In~\cite{carvalho2024decreasing} similar results were proved using Gröbner basis theoretical arguments. 

We next turn our attention to the problem of estimating relative generalized Hamming weights. Let $D \subseteq {\mathbb{F}}_{q^s}^n$ be any subspace, of dimension say $m$ and consider $\rho(D)=\{\gamma_1, \ldots , \gamma_m\} \subseteq H^\ast (Q)$ as in Section~\ref{secgeneral}. We introduce the function $\iota : H^\ast(Q) \rightarrow \{(i,j) \mid 0 \leq i < u(q-1)+1 {\mbox{ and }} 0 \leq j \leq q^{s-1}-1\}$ given by $\iota (\lambda=i q^{s-1}+j u)=(i,j)$.
In~\cite{camps2025weight} $\# {\mbox{Supp}} (D)$ was estimated for the considered curve using a Gröbner basis approach. We explain their result and demonstrate the relationship with (\ref{eqestimateD}). Let $(a,b) \in \iota (\{\gamma_1, \ldots , \gamma_m\})=\{(i_1,j_1), \ldots , (i_m,j_m)\}$ be chosen such that $a=\min \{ i_1, \ldots , i_m\}$, then the bound in~\cite[Eq.\ (3)+ Eq.\ (4)]{camps2025weight} reads
\begin{eqnarray}
\# {\mbox{Supp}} (D) &\geq &\# \{ (i,j) \mid 0 \leq i <u(q-1)+1, 0 \leq j < q^{s-1}, (i_s,j_s) \leq_p (i,j) \nonumber \\
&& {\mbox{ \ \ \ \ \ \ \ \ \ \ \ \ \ for some }} s \in \{1, \ldots , m\} {\mbox{ or }} (a+u,0) \leq_p  (i,j)\}.\label{eqvtbound}
\end{eqnarray}
Here, we used the partial ordering $\leq_p$ given by $(\alpha, \beta) \leq_p (\epsilon, \delta)$ if and only if $\alpha \leq \epsilon$ and $\beta \leq \delta$. Note, that the condition $ (a+u,0) \leq_p  (i,j)$ of course only comes into action if $a+u <u(q-1)+1$. To see that (\ref{eqvtbound}) is a consequence of~(\ref{eqestimateD}) we only need to show that $(a+u)q^{s-1}+0 u$ belongs to $(aq^{s-1},bu)+H(Q)$ whenever $0 < b$. But, $(a q^{s-1}+bu)+(0 q^{s-1}+(q^{s-1}-b)u)=(a+u)q^{s-1}+0 u$ and we are through. 

Combining~(\ref{eqvtbound}) and (\ref{eqsnabel47}) one obtains
\begin{eqnarray}
&&M_m( C_{\mathcal{L}}(G,\lambda_1Q),C_{\mathcal{L}}(G,\lambda_2Q))\nonumber \\
&\geq &\min \{ \# \{ (i,j) \mid 0 \leq i <u(q-1)+1, 0 \leq j < q^{s-1}, \iota(\gamma_s) \leq_p (i,j)   \nonumber \\
&& {\mbox{ \ \ \ \ \ \ \ \ \ \ \ \ \ \ \ \ \ \ \ \ \ for some }} s \in \{1, \ldots , m\} {\mbox{ or }} (a+u,0) \leq_p  (i,j)\}, \nonumber \\
&&{\mbox{ \ \ \ \ \ \ \ \ \ \ \ \ \ \ \ \ \ \ \ \ \ \ \ \ \ for }} s=1, \ldots , m, {\mbox{ \ }}\lambda_2 < \gamma_1 < \cdots <\gamma_m \leq \lambda_1\}.\label{eqsnabel49}
\end{eqnarray}

We continue the study of relative generalized Hamming weights but now turn to dual codes. Given a subspace $D\subseteq {\mathbb{F}}_{q^s}^n$ of dimension $m$ let $\kappa(D)=\{\eta_1,\ldots , \eta_m\}$ be as in Section~\ref{secgeneral}. Choose 
$$(a,b) \in \iota (\{\eta_1, \ldots , \eta_m\})=\{(i_1,j_1), \ldots , (i_m,j_m)\}$$ with $a$ maximal. We have
\begin{eqnarray}
\# {\mbox{Supp}} (D) &\geq &\# \{ (i,j) \mid 0 \leq i <u(q-1)+1, 0 \leq j < q^{s-1}, (i,j) \leq_p (i_s,j_s) \nonumber \\
&& {\mbox{ \ \ \ for some }} s \in \{1, \ldots , m\} {\mbox{ or }} (i,j) \leq_p  (a-u,q^{s-1}-1)\}.\label{eqaaubound}
\end{eqnarray}
We prove this by applying~(\ref{eqestimateDdual}) and by using similar arguments as above. We only need to demonstrate that $(a-u)q^{s-1}+(q^s-1)u \in (a q^{s-1}+b u) -H(Q)$ for $a\geq u$. We have $(a q^{s-1}+ bu)-(b+1)u=(a-u)q^{s-1}+(q^{s-1}-1)u$ and we are through. 

Combining~(\ref{eqaaubound}) and (\ref{eqsnabel47pkt5}) one obtains
\begin{eqnarray}
&&M_m( C^\perp_{\mathcal{L}}(G,\lambda_2Q),C^\perp_{\mathcal{L}}(G,\lambda_1Q))\nonumber \\
&\geq &\min \{ \# \{ (i,j) \mid 0 \leq i <u(q-1)+1, 0 \leq j < q^{s-1}, (i,j) \leq_p \iota (\gamma_s) \nonumber \\
&& {\mbox{ \ \ \ \ \  \ \ \ \ \ \ \ \ \ \  for some }} s \in \{1, \ldots , m\} {\mbox{ or }} (i,j) \leq_p (a-u,q^s-1)\}, \nonumber \\
&&{\mbox{ \ \ \ \ \  \ \ \ \ \ \ \ \ \ \ \ \ \ \ \ \ \ \ \ \ \ \ for }} s=1, \ldots , m, \lambda_2 < \gamma_1 < \cdots <\gamma_m \leq \lambda_1\}.\label{eqsnabel49pkt5}
\end{eqnarray}

One of the important insights from~\cite{camps2025weight} is that for so-called decreasing norm-trace codes $C_2 \subseteq C_1$ the estimate on $M_m(C_1,C_2)$ inferred from~(\ref{eqvtbound}) becomes sharp. A code of dimension $k$ is called a decreasing norm-trace code if it equals the span of functions $g_i$, $i=1, \ldots , k$ evaluated at $P_1, \ldots , P_n$ where we have $g_i =x^{\alpha_i} y^{\beta_i}$  with $0 \leq \alpha_i < (q-1)u+1, 0 \leq \beta_i < q^{s-1}$, and where for any $0 \leq \alpha \leq \alpha_i$ and $0 \leq \beta \leq \beta_i$ there exists a $j \in \{1, \ldots , k\}$ such that $g_j=x^\alpha y^\beta$. Hence, the one-point algebraic geometric codes related to the extended norm-trace curve are of this type. As a consequence of the sharpness in case of decreasing codes we can conclude that in the case of the extended norm-trace curve the right hand side of~(\ref{eqestimateD}) and the right hand side (\ref{eqvtbound}) are identical (this could alternatively have been proved directly by using~\cite[Lem.\ 5.15]{handbook}). 

In~\cite[Thm.\ 5.3]{camps2025weight} the dual of decreasing norm-trace code is shown to be equivalent to another decreasing norm-trace code. We leave it for the reader to inspect that the estimate one obtains by applying such correspondence to the relative generalized Hamming weights of a pair of duals of nested decreasing norm-trace codes is in fact the same as one gets by applying~(\ref{eqaaubound}) directly. From the above we conclude that both (\ref{eqsnabel49}) and (\ref{eqsnabel49pkt5}) are sharp. 

 \section{Schemes with a Second Layer of Security}\label{secsecondlayer}
 As mentioned in the previous section it is proved in~\cite{camps2025weight} that all their bounds are sharp in that for any $ \gamma_1 < \cdots < \gamma_m$ in $H^\ast(Q)$ there exists a corresponding space $D$ of dimension $m$ satisfying $\rho(D)=\{\gamma_1, \ldots , \gamma_m\}$ and with equality in~(\ref{eqvtbound}). Moreover, among such spaces there exist some which can be written 
 $$D= {\mbox{Span}}_{{\mathbb{F}_{q^s}}}\{ (f_1(P_1), \ldots , f_1(P_n)), \ldots , (f_m(P_1), \ldots , f_m(P_n))\}$$
 with  $-\nu_Q(f_i)=\gamma_i$, $i=1, \ldots , m$ 
  and each element $f_i$ being a product of linear factors. Our treatment of ramp secret sharing shall rely heavily on such observations.

With the aim of establishing  ramp secret sharing schemes with maximum non-$i$-qualifying sets allowing for a second layer of security we now revisit some of the results in~\cite{camps2025weight}. Here, by the first level of security we refer to the parameters $t_1, \ldots , t_\ell$ and $r_1, \ldots , r_\ell$, and the second layer means that we have systematic families of sets of maximal size of participants that cannot all be ignored when $i$ $\log_2(q)$ bits of information $i \in \{1, \ldots , \ell \}$ are to be retrieved. The following theorem is a combination of~\cite[Lem.\ 3.1]{camps2025weight} and Case 1.1 in the proof of~\cite[Thm.\ 3.2]{camps2025weight} adapted to our language.  Well-structured maximum non-$i$-qualifying sets derived from the theorem are discussed in the subsequent remark, corollary and examples.
 \begin{theorem}\label{propsomething}
 Consider pairwise different $\gamma_1, \ldots , \gamma_w \in H^\ast (Q)$ and write
 $$\iota ( \{ \gamma_1, \ldots , \gamma_w\})=\{(a_1,b_1), \ldots , (a_w,b_w)\}.$$
 Assume  $a_1<\cdots <a_w$ and $b_w < \cdots <b_1$ and that $a_w-a_1 < u$. Assume $a_1\leq u(q-2)+1$ Then the right hand side of~(\ref{eqvtbound}) reads
 \begin{equation}
 n-\big( a_1 q^{s-1}+b_w u+\sum_{i=1}^{w-1}(a_{i+1}-a_i)(b_i-b_w) \big).
 \label{eqbrugesmaaske}
 \end{equation}
 The following functions $f_1, \ldots , f_w \in R$ satisfy $-\nu_Q(f_j)=a_j q^{s-1}+b_j u$, $j=1, \ldots , w$
   and for 
 $$
 D={\mbox{Span}}_{{\mathbb{F}}_{q^s}}\{(f_1(P_1),\ldots ,f_1(P_n)), \ldots , (f_w(P_1), \ldots , f_w(P_n))\}
 $$
 equality holds in~(\ref{eqvtbound}) meaning that  $\# {\mbox{Supp}}(D)$ is equal to the right hand side of~(\ref{eqbrugesmaaske}). Let $i^\prime \in \{1, \ldots , q-1\}$ and choose 
 \begin{equation}
 \alpha_1, \ldots , \alpha_{a_1} \in \Gamma_0^{(1)} \times \Gamma_1^{(1)} \times \cdots \times  \Gamma_{i^\prime -1}^{(1)} \times \Gamma_{i^\prime + 1}^{(1)} \times \cdots \times \Gamma_{q-1}^{(1)}. \label{eqnuhenvisetil}
 \end{equation}
 Enumerate $\Gamma_{i^\prime}^{(1)}=\{\alpha_1^\prime, \ldots , \alpha_u^\prime\}$ and $\Gamma_{i^\prime}^{(2)}=\{\beta_1, \ldots , \beta_{q^{s-1}}\}$. 
 Finally for $j=1, \ldots , w$ define
 \begin{eqnarray}
 f_j=\prod_{i=1}^{a_1}(x-\alpha_i)\prod_{i=1}^{a_j-a_1}(x-\alpha_i^\prime)\prod_{i=1}^{b_j}(y-\beta_i). \label{eqf}
 \end{eqnarray}
 \end{theorem}
 \begin{proof}
See~\cite[Lem.\ 3.1]{camps2025weight} and the first part of the proof of~\cite[Thm.\ 3.2]{camps2025weight}. 
 \end{proof}

To make Theorem~\ref{propsomething} operational in connection with deriving second layer of security we shall need the following lemma which is new. 
 \begin{lemma}
Consider $\gamma_1 < \cdots < \gamma_w$ in $H^\ast(Q)$. Assume $\gamma_w-\gamma_1 < \min \{u,q^{s-1}\}$ and write 
\begin{equation}
\{\iota (\gamma_1), \ldots , \iota(\gamma_w)\}=\{ (a_1,b_1),\dots , (a_w,b_w)\}. \label{eqmercedesa}
\end{equation}
The enumeration on the right hand side of (\ref{eqmercedesa}) can be done in such a way that
\begin{equation}
a_1 < \cdots <a_w  {\mbox{ and }} b_w < \cdots <b_1. \label{eqmedogsaa} 
\end{equation}
 It holds that $a_w-a_1 < u$. Under the additional condition $\gamma_1 \leq (u (q-2)+w)q^{s-1}$ it holds that $a_1 \leq u(q-2)+1$.
 \end{lemma}
 \begin{proof}
 The assumption $\gamma_j-\gamma_i < \min \{u,q^{s-1}\}$ for $1 \leq i<j \leq w$ implies~(\ref{eqmedogsaa}). Aiming for a contradiction assume $a_w-a_1 \geq u$, and recall that by the very definition of the function $\iota$ we have $b_1-b_w  \leq q^{s-1}-1$. We obtain 
 $$(a_w q^{s-1}+b_wu)-(a_1q^{s-1}+b_1u)\geq u q^{s-1}-(q^{s-1}-1)u=u$$
 which by assumption is impossible. Finally, the additional condition in combination with $\gamma_w-\gamma_1 < \min \{u,q^{s-1}\}$ ensures that $\iota (\gamma_i) \not{\! \! \!>_p} (u(q-2)+w,0)$, $i=1, \ldots , w$. But, the $a_i$s constitute a strictly increasing sequence and we are through (to avoid confusion be aware that $\iota(\gamma_1)$ needs not be equal to $(a_1,b_1)$).
 \end{proof}
 
 \begin{corollary}\label{cornuskerdet}
 Consider $\lambda_2 < \lambda_1 $ with $\lambda_2+1, \lambda_1 \in H^\ast(Q)$ and $\lambda_1-(\lambda_2+1) < \min\{u,q^{s-1}\}$ and $\lambda_2 < (u (q-2)+w)q^{s-1}$ and consider the nested codes $C_2=C_{\mathcal{L}}(G,\lambda_2Q) \subseteq C_{\mathcal{L}}(G,\lambda_1Q)=C_1$ the co-dimension $\ell$ of which equals $\# \big( H^\ast (Q) \cap \{\lambda_2+1, \ldots , \lambda_1\}\big)$. For $w \in \{1, \ldots , \ell\}$ let $\{\gamma_1, \ldots , \gamma_w \} \subseteq H^\ast (Q)$ with $\lambda_2 < \gamma_1 < \cdots < \gamma_w \leq \lambda_1$ be such that the minimum value is attained in~(\ref{eqvtbound}) among all possible choices of $\{ \gamma_1, \ldots \gamma_w\}$ of this form. I.e.\ the right hand side of (\ref{eqvtbound}) for the given $\gamma_1, \ldots , \gamma_w$ achieves the value of $M_w(C_1,C_2)$. The set of positions corresponding to common roots of related functions $f_1, \ldots  ,f_w$ as in Theorem~\ref{propsomething} constitutes a maximum non-$(\ell-w+1)$-qualifying set. In other words, by leaving out all participants corresponding to non-common roots of these functions the remaining participants cannot detect $(\ell-w+1)$ $\log_2(q)$ bits of information.
 \end{corollary}
 Corollary~\ref{cornuskerdet} in combination with the particular structure of $\{P_1, \ldots , P_n\}$ as well as the particular structure of each of $f_1, \ldots , f_w$ imposes a second layer of security in a large family of ramp secret sharing schemes defined from nested one-point algebraic geometric codes over the extended norm-trace curves.  Such schemes have families of well-structured sets of participants who cannot all be left out if  $(\ell-w+1)$ $\log_2(q)$ bits of information are to be retrieved. This is discussed in the following remark.
 
 \begin{remark}\label{remromstang}
 Recall, that the affine variety of the extended norm-trace curve equals the disjoint union $\cup_{v =0}^{q-1} A_v$ where $A_i =\Gamma_i^{(1)} \times \Gamma_i^{(2)}$, and where $\#\Gamma_i^{(2)}= q^{s-1}$ for $i=0,\ldots , q-1$, where $\# \Gamma_0^{(1)}=1$ and where for $i=1, \ldots , q-1$ $\# \Gamma_i^{(1)}=u$. To illustrate the second layer of security assume in the following that we have an organization with  $q-1$ large departments each having  $u q^{s-1}$ members. Say, $A_i$ corresponds to a large department $i$, $i=1, \ldots , q-1$. Assume further that we have a single small department of size $q^{s-1}$. This department corresponds to $A_0$. If $a_1 < u(q-2)+1$ then one may choose $\alpha_1, \ldots , \alpha_{a_1}$ in Theorem~\ref{propsomething} in such a way that the common roots of the $f_1, \ldots , f_w$ in~(\ref{eqf}) correspond to the disjoint union of the following sets:  $\lfloor \frac{a_1}{u}\rfloor$ entire large departments $A_{i_1} \cup \cdots \cup A_{i_{\lfloor a_1/u \rfloor}}$, and for some $i^{\prime \prime} \neq i^\prime$ both belonging to $\{1, \ldots , q-1\}\backslash \{i_1, \ldots , i_{\lfloor a_1/u \rfloor}\}$ a subset $S \times \Gamma_{i^{\prime \prime}}^{(2)}\subseteq A_{i^{\prime \prime}}$, where $S \subseteq \Gamma_{i^{\prime \prime}}^{(1)}$, $\#S=a_1-\lfloor \frac{a_1}{u} \rfloor$ and finally a subset of $A_{i^\prime}$ of size equal to $b_wu +\sum_{i=1}^{w-1}(a_{i+1}-a_i)(b_i-b_w)$ with the form of some (possibly irregular) staircase. We call the latter subset ``the set locally induced by $\gamma_1, \ldots , \gamma_w$'' (or ``the locally induced set'' for short). Observe, that $i^\prime$ plays the exact same role as in Theorem~\ref{propsomething}, whereas $i^{\prime \prime}$ is introduced to reflect a particular systematic way for choosing the elements in the left hand side of~(\ref{eqnuhenvisetil}). The subset $S \times \Gamma_{i^{\prime \prime}}^{(2)}$ as well a the locally induced set can be given their own meaning. This is done by dividing $A_{i^{\prime \prime}}$ and $A_{i^\prime}$, respectively, into $u$ horizontal levels each containing $q^{s-1}$ elements. Here, we enumerate the levels according to $\Gamma_{i^{\prime \prime}}^{(1)}$ and $\Gamma_{i^\prime}^{(1)}$, respectively. Similarly, we can in an obvious way divide $A_{i^{\prime \prime}}$ and $A_{i^\prime}$, respectively, into $q^{s-1}$ vertical levels each containing $u$ members. Hence, for instance $S \times \Gamma_{i^{\prime \prime}}^{(2)}$ consist of $a_1-\lfloor \frac{a_1}{u} \rfloor$ horizontal levels. Note, that for fixed $A_{i^{\prime \prime}}$ there are ${u}\choose{a_1-\lfloor a_1/u \rfloor }$ possibilities for that.
If $a_1=u(q-2)+1$ then there is in the set of common roots no set $S \times \Gamma_{i^{\prime \prime}}^{(2)}$ and we need to include the entire set $A_0$. Returning to the case $a_1 < u(q-2)+1$ we may of course also include $A_0$ in the zero-set, which then causes a minor change in how we may include the other departments.
The many different ways one can define $f_1, \ldots , f_w$ given fixed $\gamma_1, \ldots ,  \gamma_w$ provides us with many different very systematic patterns of common roots, and similarly, of course, of the same number of different very systematic patterns of non-common roots, the latter being those participants if all left out, the remaining participants cannot detect $\ell-w+1$ $\log_2(q)$ bits of information. This concludes the remark.
 \end{remark}

 We illustrate the idea with a couple of examples where the first describes a situation where the set of common roots (and therefore also the set of non-common roots) are as simple as can possibly be.

 \begin{example}
Let $q \geq 3$ be a prime-power, $s \geq 2$, and $u \geq 2$. Consider codes $C_2=C_{\mathcal{L}}(G,\lambda_2 Q) \subseteq C_{\mathcal{L}}(G,\lambda_1 Q)= C_1 \subseteq {\mathbb{F}}_{q^s}^n$ of co-dimension $\ell= 1$ defined by $\iota (\lambda_1)=(\tau u,0)$ with $\tau \in \{1, \ldots , q-2\}$. The conditions of Corollary~\ref{cornuskerdet} are clearly satisfied and by applying~(\ref{eqtm}), (\ref{eqrm}), (\ref{eqbrugesmaaske}) and (\ref{eqsnabel49pkt5}) we obtain
\begin{eqnarray}
t&=&M_1(C_2^\perp, C_1^\perp)-1 = (\tau -1)u q^{s-1}+q^{s-1}+u-1 \nonumber \\
r&=&n-M_1(C_1,C_2)+1 =\tau u q^{s-1}+1 \nonumber \\
r-t&=&(q^{s-1}-1)(u-1)+1 \nonumber 
\end{eqnarray}
In the following we use the language of Remark~\ref{remromstang}. 
Now for all $1 \leq i_1 < \cdots < i_\tau \leq q-1$ the set $A_{i_1}\cup \cdots \cup A_{i_\tau}$ 
  constitutes a maximum non-$1$-qualifying set. Hence, by leaving out all members of $(q-1)-\tau$ large departments and the small department one does not obtain any information. In particular if $\tau=q-2$ one cannot leave out an entire large department in combination with the small department if one wants to generate information. 
  Next let $\iota (\lambda_1)=(\tau u +1, 0)$ with $\tau \in \{1, \ldots , q-2\}$. We obtain
  $$r=(\tau u +1)q^{s-1}+1$$
  with the same value of $r-t$ as before. By leaving out all members of any set of $(q-1)-\tau $ large departments one does not obtain any information. And by leaving out any $(q-1)-\tau-1$ large departments, the small department as well as $S\times \Gamma_i^{(2)}$ where $A_i$ $i\in \{1, \ldots , q-1\}$ is not one of the already left out large department and where $S \subseteq \Gamma_i^{(1)}$ is of size equal to $u-1$ one neither obtains any information. The complementary sets are maximum non-$1$-qualifying. In particular if $\tau =q-2$ one cannot leave out an entire large department and obtain any information. Similarly, one cannot leave out the small department and $S \times \Gamma_i^{(2)}$, $i \in \{1, \ldots , q-1\}$ where $S$ is given as above.
 \end{example}
 
 \begin{example}
 Consider the Hermitian curve $x^{q+1}-y^q-y$ over ${\mathbb{F}}_{q^2}$. We have $u=q+1$, $s=2$, and $q^{s-1}=q$. Let $\lambda_2+1=a_\ell q$ where $q-1 \leq a_\ell < u (q-2)+1=(q+1)(q-2)+1$. Let $\ell = q$ and $\lambda_1= \lambda_2+\ell$. Then $\lambda_2+1, \ldots , \lambda_2+\ell$ all belong to $H^\ast (Q)$. Moreover, we have $\iota (\lambda_2+1+i)=(a_\ell -i,i)$ for $i=0, \ldots , \ell-1$. Hence, with $C_2\subseteq C_1$ as in Corollary~\ref{cornuskerdet} all conditions therein are satisfied. Consider the related ramp secret sharing scheme. We have
 \begin{eqnarray}
 t_1&=& M_1(C_2^\perp , C_1^\perp)-1=(a_\ell +1)+(q-1)(a_\ell -q)-1\nonumber \\
 t_\ell&=& M_\ell (C_2^\perp, C_1^\perp ) -1 =\big( (a_{\ell}-(q-1) )q +\sum_{i=0}^{\ell -1} i\big)  -1 = (a_{\ell}-q+1)q+\frac{(q-1)q}{2}-1\nonumber \\
 r_1&=&n-M_\ell(C_1,C_2) +1=a_\ell q-(\sum_{i=1}^\ell (i-1))+1=a_\ell q  -\frac{(q-1)q}{2} +1 \nonumber \\
 r_\ell&=& n-M_{1}(C_1,C_2)+1=n-(n -\lambda_1) +1=(a_\ell +1)q.\nonumber
 \end{eqnarray}
Using the approach described in Remark~\ref{remromstang} we have maximum non-$1$-qualifying sets of the following form. Namely,  the disjoint union of $\lfloor \frac{a_\ell - (q-1)}{q+1}$ large departments, the following subset of a large department  $S \times \Gamma_{i^{\prime \prime }}\subseteq A_{i^{\prime \prime}}$, with $\#S= a_\ell -(q-1) - \lfloor (a_\ell-(q-1))/(q+1)\rfloor$ and finally a locally induced set (subset of $A_{i^{\prime}}$)  which has the form of a staircase each step having height equal to $1$. We obtain maximum non-$\ell$-qualifying sets as above, but with the latter locally induced set being replaced by $q-1$ vertical levels each containing $q+1$ elements. One, of course can also investigate non-$i$-qualifying sets with $1 < i < \ell=q$, but we shall refrain from that in the present exposition.
 \end{example}

 \begin{example}
 In this example we consider $q=4$, $s=3$ and $u=\frac{(q^s-1)}{(q-1)3}=7$. We obtain codes over ${\mathbb{F}}_{64}$ of length $n=(u(q-1)+1 ) q^{s-1}= 352$. We have $88, 89 \notin H^\ast(Q) \subseteq \langle 7, 16 \rangle$, but $87, 90, 91, 92 \in H^\ast(Q)$. We shall consider a pair of nested codes of co-dimension $\ell =3$, namely, $C_2\subseteq C_1$ where $C_2 =C_{\mathcal{L}}(G,87Q)$ and $C_1=C_{\mathcal{L}}(G,92Q)$ for which we note that all conditions of Corollary~\ref{cornuskerdet} are satisfied. We have $\iota (90)=(3,6)$, $\iota (91)=(0,13)$, and $\iota (92)=(4,4)$ Applying~(\ref{eqbrugesmaaske}) we calculate
 \begin{eqnarray}
M_1(C_1,C_2)&=& \min \{ 262, 261, 260\} = 260 \label{eqjkj1} \\
M_2(C_1,C_2)&=&\min \{ 289, 276, 274\} = 274\label{eqjkj2} \\
M_3(C_1,C_2)&=& 295, \label{eqjkj3}
\end{eqnarray}
 where the three values on the right hand side of (\ref{eqjkj1}) each corresponds to a calculation concerning one of the values $90, 91, 92$ the minimum being attained for $92$. The three values on the right hand sides of (\ref{eqjkj2}) each corresponds to a calculation concerning a pair of such values, the minimum being attained for $\{90, 92\}$. Finally, (\ref{eqjkj3}) corresponds to a calculation concerning all three numbers $90, 91, 92$ at the same time. 
 To calculate the relative weights of the dual codes we apply~(\ref{eqsnabel49pkt5}) directly. We obtain
\begin{eqnarray}
M_1(C_2^\perp , C_1^\perp) &=& \min \{ 14, 28, 25 \} =14 \label{eqjkj4} \\
M_2(C_2^\perp, C_1^\perp)&=&\min \{ 35, 34, 33\} =33 \label{eqjkj5} \\
M_3(C_2^\perp,C_1^\perp)&=&  40. \label{eqjkj6}
\end{eqnarray}
 From the above in combination with~(\ref{eqtm}) and (\ref{eqrm}) we obtain
 $t_1=259$, $t_2=273$, $t_3=294$, $r_1=313$, $r_2=320$, and $r_3=339$. To detect maximum non-$1$-qualifying sets we should apply Remark~\ref{remromstang} to $\{\iota(90) \iota(91), \iota (92)\}$. Here, $a_1=0$ and therefore all the sets we obtain consist of a locally induced set and nothing else. So leaving out the entire set of members from all departments but one large department, there are limits as to which patterns of members from the remaining department that could be left out. To detect maximum non-$2$-qualifying sets we should consider $\{ \iota (90), \iota(92)\}$. Here, $a_1=3$. Hence, we obtain sets consisting of a Cartesian product $S \times \Gamma_{i^{\prime \prime}}^{(2)}$, $\#S=3$ in combination with a locally induced set. Finally, to detect maximum non-$1$-qualifying sets we should consider $\{ \iota ( 92)\}$. Here, $a_1=4$, and we obtain sets that are the union of a Cartesian product $S \times \Gamma_{i^{\prime \prime}}^{(2)}$, $\#S=4$ and a locally induced set. 
 \end{example}

\begin{remark}\label{remsidste}
Revisiting  Theorem~\ref{propsomething} and the conditions therein keep the assumption $\gamma_1 < \cdots < \gamma_w $ in $H^\ast(Q)$ again with $\iota (\{ \gamma_1 , \ldots , \gamma_w\})=\{(a_1,b_1), \ldots , (a_w,b_w)\}$ and with the condition that $a_1 < \cdots < a_w$ and $b_w < \cdots < b_1$. But instead of requiring  $a_1 \leq u(q-2)+1$ assume now $u(q-2)+1 < a_1$. Inspecting~(\ref{eqvtbound}) we see that the last option does not come into action as $u(q-1)+1 < a_1+u$ which does not correspond to the first coordinate of any $\iota (\delta)$ where $\delta \in H^\ast(Q)$. Hence, the right hand side of~(\ref{eqvtbound}) becomes 
\begin{eqnarray}
\# \{ (i,j) \mid 0 \leq i < u(q-1)+1, 0 \leq j \leq q^{s-1}-1,  (a_s,b_s) <_p (i,j) {\mbox{ \ \ \ \ }}\nonumber \\
{\mbox{ for some }} s\in \{1, \ldots , w\} \}.\nonumber 
\end{eqnarray}
But then the situation is similar to that of monomial-Cartesian codes regarding parameters of primary codes (see~\cite[Eq.\ (5)]{geil2026considerate}) and the maximum non-$i$-qualifying sets have a similar structure, however with a little less freedom of choice~(see~\cite[Thm.\ 6]{geil2026considerate}). The advantage of using codes from extended norm-trace curves for such high value of $a_1$ does not lie in the reconstruction numbers, nor the structure of maximum non-$i$-qualifying sets. Rather, it is the privacy numbers that are improved.
\end{remark}
We illustrate Remark~\ref{remsidste} with a final example.
\begin{example}
We first consider nested one-point algebraic geometric codes over ${\mathbb{F}}_{16}$ using the Hermitian curve $X^5-Y^4-Y$. We have $\iota (66)=(14,2)$ and $\iota (67)=(13,3)$ both $66$ and $67$ belonging to $H^\ast(Q)$ as $$\iota ( H^\ast (Q))=\{(i,j) \mid 0 \leq i < 16, 0 \leq j < 4\}.$$ Hence, the co-dimension of $C_2=C_{\mathcal{L}}(G,65Q)\subseteq C_{\mathcal{L}}(G,67Q)=C_1$ is $\ell=2$. We have 
\begin{eqnarray} 
M_2(C_1,C_2)&=&\# \{ (14,2), (13,3), (14,3), (15,2), (15,3)\}=5 \nonumber \\
M_1(C_1,C_2)&=&\min \{ \# \{ ((13,3), (14,3), (15,3)\}, \# \{(14,2), 14,3), (15,2), (15,3)\} \}\nonumber \\
&=&3 \nonumber 
\end{eqnarray}
We have maximum non-$1$-qualifying sets of the form: all $64$ points in the entire affine variety except $$\{ (\alpha_1^\prime, \beta_1), (\alpha_2^\prime, \beta_1), (\alpha_1^\prime, \beta_2), (\alpha_2^\prime, \beta_2), (\alpha_3^\prime, \beta_2)\}$$
where we use the notation from Theorem~\ref{propsomething} with the arbitrary enumeration from there. We have maximum non-$2$-qualifying sets of the form: all $64$ points except $$\{ (\alpha_1^\prime, \beta_1), (\alpha_2^\prime, \beta_1), (\alpha_3^\prime, \beta_1)\}.$$ Turning to comparable monomial-Cartesian codes we consider as point set (affine variety) now ${\mathbb{F}}_{16} \times S_2=\{ P_1^\prime, \ldots , P_{64}^\prime\}$ i.e.\ $\# S_2 =4$. We write ${\mathbb{F}}_{16}=\{ \delta_1, \ldots , \delta_{16}\}$ and $S_2=\{ \epsilon_1, \ldots , \epsilon_4\}$ and consider an evaluation map  ${\mbox{ev}}: {\mathbb{F}}_{16}[X,Y] \rightarrow {\mathbb{F}}_{16}^{64}$ given by ${\mbox{ev}}(F)=(F(P_1^\prime), \ldots , F(P_{64}^\prime))$. Define 
\begin{eqnarray}
C_2^\prime &=&{\mbox{Span}}_{\mathbb{F}_{16}} \{ {\mbox{ev}} (X^iY^j) \mid 0 \leq i < 16, 0 \leq j < 4, 4i+5j \leq 65\} \subseteq \nonumber \\
C_1^\prime &=&{\mbox{Span}}_{\mathbb{F}_{16}} \{ {\mbox{ev}} (X^iY^j) \mid 0 \leq i < 16, 0 \leq j < 4, 4i+5j \leq 67\}.\nonumber 
\end{eqnarray} 
The relative parameters are the same as before (see~\cite[Eq.\ (5)]{geil2026considerate}), but now we have that
the entire affine variety (still of size $64$) minus $$\{(\delta_1, \epsilon_1), (\delta_2,\epsilon_1), (\delta_1, \epsilon_2), (\delta_2,\epsilon_2), (\delta_3,\epsilon_2)\}$$ is a maximum non-$1$-qualifying set (see~\cite[Thm.\ 6]{geil2026considerate}). Similarly, the entire affine variety minus $$\{(\delta_1,\epsilon_1), (\delta_2, \epsilon_1), (\delta_3, \epsilon_1)\}$$ is a maximum non-$2$-qualifying set. Due to the many ways one can enumerate the elements of ${\mathbb{F}}_{16}$ and similarly enumerate the elements of $S_2$ the second layer of security becomes much stronger for the secret sharing scheme defined from the latter pair of codes. 
\end{example}

In this paper we only considered the case of not too large co-dimension of two one-point algebraic geometric codes $C_2 \subseteq C_1$ imposing a situation where $a_1 < \cdots < a_w$ and $b_w < \cdots < b_1$. It is also possible to consider higher co-dimension for which we refer the reader to employ the last part of~\cite[Prf. of Thm.\ 3.2]{camps2025weight}.  
 
 \section{Concluding Remarks}\label{secconclusion}
The second layer of security resulting from families of maximum non-$i$-qualifying sets of systematic form was first considered in~\cite{geil2026considerate} in the case of monomial-Cartesian codes. In the present work it was analyzed for schemes defined from nested one-point algebraic geometric codes defined over extended norm-trace curves. There should be other algebraic codes for which the resulting schemes have a second layer of security. In this paper we demonstrated that the ``footprint-like approach'' applied in~\cite{camps2025weight} is merely a clever application of the enhanced Goppa bound in~\cite{geil2014relative} than a competing method. A similar remark could be made concerning~\cite{normtrace} which, however, predates~\cite{geil2014relative}. In the opinion of the author of this contribution, indeed the ``footprint-like approach'' may be beneficial, but merely in connection with affine variety codes not defined from $\cup_{m=0}^\infty {\mathcal{L}}(mQ)$ where $Q$ is a rational place in an algebraic function field of transcendence degree $1$. The literature contains several examples of such studies including codes defined from other structures related to function fields of transcendence degree $1$ than the above union of ${\mathcal{L}}$-spaces, e.g.~\cite{geilozbudak}.


\begin{thebibliography}{10}
\providecommand{\url}[1]{\texttt{#1}}
\providecommand{\urlprefix}{URL }
\providecommand{\doi}[1]{https://doi.org/#1}

\bibitem{MR1737936}
Barbero, A.I., Munuera, C.: The weight hierarchy of {H}ermitian codes. SIAM J.
  Discrete Math.  \textbf{13}(1),  79--104 (electronic) (2000).

\bibitem{Blakley}
Blakley, G.R., Meadows, C.: Security of ramp schemes. In: Advances in
  cryptology ({S}anta {B}arbara, {C}alif., 1984), Lecture Notes in Comput.
  Sci., vol.~196, pp. 242--268. Springer, Berlin (1985)

\bibitem{bras2008duality}
Bras-Amor{\'o}s, M., O'Sullivan, M.E.: Duality for some families of correction
  capability optimized evaluation codes. Adv.\ Math.\ Commun.  \textbf{2}(1),
  15--33 (2008)

\bibitem{camps2025weight}
Camps-Moreno, E., L{\'o}pez, H.H., Matthews, G.L., San-Jos{\'e}, R.: The weight
  hierarchy of decreasing norm-trace codes. Des.\ Codes Cryptogr.
  \textbf{93}(7),  2873--2894 (2025)

\bibitem{carvalho2024decreasing}
Carvalho, C., L{\'o}pez, H.H., Matthews, G.L.: Decreasing norm-trace codes.
  Des.\ Codes Cryptogr.  \textbf{92}(5),  1143--1161 (2024)

\bibitem{cascudo2013bounds}
Cascudo, I., Cramer, R., Xing, C.: Bounds on the threshold gap in secret
  sharing and its applications. IEEE Trans.\ Inform.\ Theory  \textbf{59}(9),
  5600--5612 (2013)

\bibitem{ChenCramer}
Chen, H., Cramer, R., Goldwasser, S., de~Haan, R., Vaikuntanathan, V.: Secure
  computation from random error correcting codes. In: Advances in
  cryptology---{EUROCRYPT} 2007, Lecture Notes in Comput. Sci., vol.~4515, pp.
  291--310. Springer, Berlin (2007)

\bibitem{csirmaz2009ramp}
Csirmaz, L.: Ramp secret sharing and secure information storage. Preprint
  (2009)

\bibitem{MR3782266}
Galindo, C., Geil, O., Hernando, F., Ruano, D.: Improved constructions of
  nested code pairs. IEEE Trans. Inform. Theory  \textbf{64}(4, part 1),
  2444--2459 (2018)

\bibitem{normtrace}
Geil, O.: On codes from norm-trace curves. Finite Fields Appl.  \textbf{9}(3),
  351--371 (2003)

\bibitem{geil2026considerate}
Geil, O.: Considerate ramp secret sharing. Des.\ Codes Cryptogr.
  \textbf{94}(3), ~49 (2026)

\bibitem{geil2014relative}
Geil, O., Martin, S., Matsumoto, R., Ruano, D., Luo, Y.: Relative generalized
  {H}amming weights of one-point algebraic geometric codes. IEEE Trans. Inform.
  Theory  \textbf{60}(10),  5938--5949 (2014)
  
  \bibitem{geilozbudak}
Geil, O., \"{O}zbudak, F.: On affine variety codes from the {K}lein quartic.
  Cryptogr. Commun.  \textbf{11}(2),  237--257 (2019)
 
\bibitem{handbook}
H{\o}holdt, T., van Lint, J.H., Pellikaan, R.: Algebraic geometry codes. In:
  Pless, V.S., Huffman, W.C. (eds.) Handbook of Coding Theory, vol.~1, pp.
  871--961. Elsevier, Amsterdam (1998)

\bibitem{janwa2016parameters}
Janwa, H., Pi{\~n}ero, F.L.: On parameters of subfield subcodes of extended
  norm-trace codes. arXiv preprint arXiv:1604.05777  (2016)

\bibitem{kurihara}
Kurihara, J., Uyematsu, T., Matsumoto, R.: Secret sharing schemes based on
  linear codes can be precisely characterized by the relative generalized
  {H}amming weight. IEICE Trans. Fundamentals  \textbf{E95-A}(11),  2067--2075
  (2012)

\bibitem{matthews2022norm}
Matthews, G.L., Murphy, A.W.: Norm-trace-lifted codes over binary fields. In:
  Proceedings of 2022 IEEE International Symposium on Information Theory (ISIT). pp.
  3079--3084. IEEE (2022)
  
\bibitem{patel2024efficient}
Patel, S., Persiano, G.,  Seo, J. Y.,  Yeo, K.: Efficient secret sharing for large-scale applications. In: Proceedings of the 2024 ACM SIGSAC Conference on Computer and Communications Security. pp. 3065--3079 (2024)

\bibitem{pellikaanexists}
Pellikaan, R.: On the existence of order functions. J. Statist. Plann.
  Inference  \textbf{94}(2),  287--301 (2001)

\bibitem{stichtenothhermitian}
Stichtenoth, H.: A note on {H}ermitian codes over {GF($q^2$)}. IEEE Trans.
  Inform. Theory  \textbf{34}(5),  1345--1348 (1988)

\bibitem{Yamamoto}
Yamamoto, H.: Secret sharing system using {$(k,L,n)$} threshold scheme.
  Electron. Comm. Japan Part I Comm.  \textbf{69}(9),  46--54 (1986)

\bibitem{yanghermitian}
Yang, K., Kumar, P.V.: On the true minimum distance of {H}ermitian codes. In:
  Coding theory and algebraic geometry, pp. 99--107. Springer (1992)

\end{thebibliography}
\end{document}